\documentclass[11pt]{article}
\usepackage{fullpage}
\usepackage{amsthm,amssymb,comment,ordinalref,enumerate,xspace,cite}
\usepackage{multirow,microtype,fullpage}
\usepackage[bookmarks,pdfpagemode=UseOutlines,citecolor=blue,bookmarksnumbered=true]{hyperref}

\newcommand{\zo}{{\{0,1\}}}

\newcommand{\PAR}{{\sf {PAR}}}

\newcommand{\dist}{{\mathrm{dist}}}

\newtheorem{theorem}{Theorem}
\newtheorem{corollary}[theorem]{Corollary}

\newtheorem{lemma}[theorem]{Lemma}

\newtheorem{claim}[theorem]{Claim}

\title{The non-adaptive query complexity of testing $k$-parities}
\author{Harry Buhrman\thanks{CWI and University of Amsterdam, The Netherlands. Partially supported by the European Commission under the project QCS (Grant No.~255961).}
\and
David Garc\'ia-Soriano\thanks{Yahoo! Research, Barcelona, Spain. Part of this work completed while at CWI Amsterdam.}
\and
Arie Matsliah\thanks{Google, Mountain View, USA.}
\and
Ronald de Wolf\thanks{CWI and University of Amsterdam, The Netherlands. Partially supported by a Vidi grant from the Netherlands Organization for Scientific Research (NWO), and by the European Commission under the project QCS (Grant No.~255961).}
}

\begin{document}

\maketitle

\begin{abstract}
We prove tight bounds of $\Theta(k \log k)$ queries for non-adaptively testing whether a function $f:\zo^n \to \zo$ is a $k$-parity or far from any $k$-parity. 
The lower bound combines a recent method of Blais, Brody and Matulef~\cite{bbm} to get lower bounds for testing from communication complexity with an $\Omega(k\log k)$ lower bound for the one-way communication complexity of $k$-disjointness.
\end{abstract}

\section{Introduction}

A \emph{parity} is a function $f:\zo^n \to \zo$ that can be written as $f(y)=\langle x, y \rangle$, the inner product (mod~2) of~$y$ with some fixed string $x$.
We also sometimes denote this function by $\chi_x$. 
It corresponds to the Hadamard encoding of the word~$x$.
We call $f$ a \emph{$k$-parity} (equivalently, a parity of \emph{size}~$k$) if $x$ has Hamming weight~$k$ and a
\emph{$\leq k$-parity} if the Hamming weight is at most~$k$. The indices of 1-bits in $x$ are called the \emph{influential variables} of~$f$. We consider the following testing problem:
\begin{quote}
Let $1\le k\le n$ be integers.
Given oracle access to a Boolean function $f:\zo^n \to \zo$, how many queries to $f$ do we need to
test (i.e.,\ determine with probability $1-\delta$, for some ``confidence parameter'' $\delta$, typically $\delta=1/3$) whether~$f$ is a $k$-parity or far from any $k$-parity?
\end{quote}
Here a function $f$ is \emph{far} from a set of functions~$G$ if for all $g\in G$,
the functions $f$ and $g$ differ on at least a constant fraction of their domain $\zo^n$.
For concreteness one can take this constant (the ``proximity parameter'') to be 1/10.
Let $\PAR^n_k$ denote the set of all $k$-parities on $n$-bit inputs, $\PAR^n_{\leq k}=\cup_{\ell\leq k}\PAR^n_\ell$, and $\PAR^n=\PAR^n_{\leq n}$.

Another way of looking at the problem is as determining, by making as few queries as possible to the Hadamard encoding of a word $x$, whether
$|x| = k$ or not.  So the task is essentially how to decide if $|x|=k$ efficiently if we can query the XOR of arbitrary subsets of the bits of $x$.%
\footnote{Decision trees where the queries are allowed to be XORs of subsets of the inputs have appeared in the literature~\cite{parity_complexity}.}

It is easy to see that deciding if the size of a parity is $k$ is the same problem as deciding if it
is $n - k$ (replace queries to $x\in \zo^n$ with queries to $1^n \oplus x$). Hence we will assume $k\leq n/2$.
For even $n$, the case $k = n / 2$ is particularly interesting because
it enables us to verify the equality between the sizes of two unknown parities $f, g \in \PAR^n$. Indeed, define a parity on
$2n$ variables by $h(x_1 x_2) = f(1^n \oplus x_1) \oplus g(x_2)$, where $x_1, x_2 \in \zo^n$; then
$h \in \PAR_n^{2 n}$ if and only if $f$ and $g$ are parities of the same size.

A related problem is deciding if a parity has size \emph{at most} $k$ (naturally, this is equivalent to deciding if the size is at least $n - k$, or at most $n - k - 1$).
Upper bounds for this task imply upper bounds for testing $k$-parities\footnote{Note that one-sided error is not preserved under this reduction.}: one can perform one test to verify the condition $|x| \le k$ and another one for $|x| \le k - 1$.
Furthermore, lower bounds for testing $k$-parities (as well as lower bounds for testing $\le k$-parities)
imply matching lower bounds for testing \emph{$k$-juntas}, i.e., functions that depend on at most~$k$ variables. This is
because one way of testing if~$f \in \PAR^n_{\le k}$ ($f \in \PAR^n_{k}$) is testing that $f$ is
linear and also a $k$-junta (but not a $(k-1)$-junta).

The first step towards analyzing the hardness of these problems was taken by Goldreich~\cite[Theorem 4]{obdds}, who proved that
testing if a linear function $f\in \PAR^n$ ($n$ even) is in $\PAR^n_{\le n/2}$ requires $\Omega(\sqrt n)$ queries.
Goldreich conjectured that the true bound should be $\Theta(n)$.
Later Blais et al.~\cite{bbm} showed that testing if a function $f$ is a $k$-parity requires $\Omega(k)$ queries.

In this paper we focus on \emph{non-adaptive} testing, where all queries to $f$ are chosen in
advance.  Our main results are tight upper and lower bounds of $\Theta(k\log k)$ non-adaptive queries
for testing whether $f$ is in, or far from, the set $\PAR^n_k$. 
Section~\ref{sec:upper} describes our upper bound and Section~\ref{sec:lower} describes our lower bound.
The lower bound combines a recent method of Blais, Brody and Matulef~\cite{bbm} to get lower bounds for testing from communication complexity, with a $\Theta(k\log k)$ bound for the one-way communication complexity of $k$-disjointness.

\paragraph{History and related work.}
This work started in 2006, when Pavel Pudlak asked HB if a lower bound on the one-way communication
complexity of $k$-disjointness was known.  He was interested in this question in order to obtain lower bounds for monotone circuits and span programs. This question prompted RdW to prove the lower bound part of
Theorem~\ref{thm:onewaycclowerbound} in October 2006.  That proof was unpublished until now, because
Pudlak managed to obtain stronger bounds without communication complexity. In April 2011, after the
Blais et al.\ paper~\cite{bbm} had come out as a preprint, we noticed that its approach could be
combined with lower bounds on one-way communication to obtain lower bounds on non-adaptive testers;
and in particular, that Theorem~\ref{thm:onewaycclowerbound} implied a $\Omega(k\log k)$ lower bound for non-adaptively testing $k$-parities. We added a matching upper bound, and these results were included in DGS's PhD thesis~\cite{garciasoriano:phd} (defended April 25, 2012) and subsequently turned into this paper.

In the mean time, in April 2009 Mihai P\v{a}tra\c{s}cu wrote a blog post~\cite{patrascu:disjblog} that stated the $\Omega(k\log k)$ communication lower bound as a folklore result, acknowledging David Woodruff for communicating this result to him  ``many years ago.''  We were not aware of this post until September 2012, when David Woodruff pointed us to it.
Neither were Dasgupta, Kumar, and Sivakumar, who independently published the $\Omega(k\log k)$ communication lower bound in~\cite{dks:sparsedisj}, nor were various other experts we had talked to earlier.

The observation that the lower-bound approach of~\cite{bbm} can be combined with \emph{one-way} communication lower bounds to obtain \emph{non-adaptive} testing lower bounds was made independently by a number of people as well, including (as an anonymous referee informed us) by ``Paul Beame, Shubhangi Saraf, and Srikanth Srinivasan in June'11 and by David Woodruff and Grigory Yaroslavtsev in May'12.''
Goldreich makes the same observation in a very recent survey~\cite{goldreich:ccproptesting}.

\section{Upper bounds}\label{sec:upper}
It is well known that testing membership in $\PAR^n_k$ can be done with $O(\log{n \choose k})$
non-adaptive queries (see, e.g.,~\cite[Proposition 5.1]{fiso_journal}), albeit with running time roughly $n^k$.
When $k$ is small ($k=n^{o(1)}$), the query and time complexities can be
improved considerably, and this is the main case of interest.

Here we prove that $O(k \log k)$ queries suffice to non-adaptively test if $f:\zo^n \to \zo$ is a
$k$-parity. We remark that the running time of our test is $n \cdot \mathrm{poly}(k)$. Furthermore, the same analysis
 can be easily generalized to prove similar upper bounds for counting the number of relevant
 variables of any $k$-junta
that is far from being a $(k-1)$-junta (in which case the influence of all its variables is bounded
 from below by  a constant). 

For $k = O(1)$ the result is easy to establish; for instance we can use the non-adaptive
isomorphism tester of Fischer et al.~\cite{juntas}, since any query complexity that is a
function of $k$ alone is fine for constant $k$. Thus below we assume $k = \omega(1)$.

First we design a tester for the special case $n = 100k^2$, and then we show how the general case
reduces to this special case. We will repeatedly invoke a simple case of the standard Chernoff bound, 
which says that if $X_1,\ldots,X_m$ are i.i.d.\ Boolean random variables each with expectation~$p$, 
then with high probability their sum is close to its expectation~$pm$:
\begin{equation}\label{eq:chernoff}
\Pr\left[ \sum_{i=1}^m X_i < (p - \epsilon)m \right] < \exp(-2\epsilon^2 m). 
\end{equation}
The basic ingredient we use is the \emph{influence test} (see~\cite[Section 3]{juntas}).
\begin{claim}\label{sec:inf_test} {\bf Influence test\ }
Let $f:\zo^n \to \zo$ be a parity function with $J \subseteq [n]$ being the set of its influential variables. There is a probabilistic procedure $I_f:\zo^n \to \zo$
that when executed on input $x \in \zo^n$ (corresponding to a set $x \subseteq [n]$) satisfies the following:
\begin{itemize}
\item $I_f$ makes at most $8$ queries to $f$;
\item if $x \cap J = \emptyset$ then $I_f$ returns 0;
\item if $x \cap J \ne \emptyset$ then $I_f$ returns 1 with probability at least $99/100$.
\end{itemize}
In other words, $I_f(x)$ is a probabilistic predicate (with one-sided error) checking if $x$ and $J$ intersect.
\end{claim}

The influence test can be made more robust, to handle functions $f$ that are only close to being parities, by increasing the query-complexity (per test) and switching to two-sided error:

\begin{claim}\label{sec:noisy_inf_test} {\bf Noisy influence test\ }
Let $f:\zo^n \to \zo$ be $1/10$-close to a parity function $g:\zo^n \to \zo$ with influential variables $J \subseteq [n]$. There is a probabilistic procedure $I^N_f:\zo^n \to \zo$ that when executed on input $x \in \zo^n$ satisfies the following:
\begin{itemize}
\item $I^N_f$ makes at most $640$ queries to $f$;
\item if $x \cap J = \emptyset$ then $I^N_f$ returns 0 with probability at least $49/50$;
\item if $x \cap J \ne \emptyset$ then $I^N_f$ returns 1 with probability at least $49/50$.
\end{itemize}
In other words, $I^N_f(x)$ is a probabilistic predicate checking if $x$ and $J$ (the influential variables of the parity function $g$ closest to $f$) intersect.
\end{claim}
\begin{proof}
We use the self-correction property of the Hadamard code: 
$$
\Pr_{y \in \zo^n}[g(x) = f(y) \oplus f(y \oplus x)] \ge 1-2\cdot \dist(f,g) \ge 4/5.
$$ 
This gives us a procedure to decode $g(x)$ with probability $\geq 4/5$ using 2~queries to~$f$.
Doing this 40 times and taking the majority value correctly decodes $g(x)$ with probability $1-1/800$ using $80$ queries 
(use the Chernoff bound~\ref{eq:chernoff} with $m=40$, $p=0.8$ and $\epsilon=0.3$). 
Hence, by the union bound, any $8$ values (for a single application of the usual influence test) can be decoded correctly with
error probability $1/100$.
Now use the tester from Claim~\ref{sec:inf_test} and observe that the overall error probability is at most $1/100+1/100=1/50$.
\end{proof}

So, prior to testing if $f$ is a $k$-parity, we test it for being a parity function with proximity
parameter $1/10$ and confidence parameter $99/100$. This can be done with a constant number of
        queries using the Linearity Test of~\cite{blr}. If this test fails then we reject; otherwise, we assume $f$ is
$1/10$-close to being a parity function, which results in only a small increase in the failure
probability.

\subsection{Testing in the case where $n=100k^2$}\label{sec:red_test}
In the following test we set $q=\frac{1000}{\rho} k \log k$, with $\rho \in (0,1]$ being a constant defined later.
\begin{itemize}
\item Draw $r_1,\ldots,r_q \in \zo^n$ at random, by setting $r_{ij}$ to $1$ with probability $\rho/k$ for each $i \in [q]$ and $j \in [n]$, independently of the others. For each $j \in [n]$, denote by $S^j \subseteq [q]$ the set of indices $i \in [q]$ with $r_{ij}=1$.
\item Compute $a_i \gets I^N_f(r_i)$ for all $i \in [q]$ with the noisy influence test of Claim~\ref{sec:noisy_inf_test}. For each $j \in [n]$ denote by $S^j_1$ the subset of $S^j$ containing indices $i$ with $a_i=1$.
\item Output the subset $\hat{J} \subseteq [n]$ containing the indices $j$ for which $|S^j_1| > \frac{3}{4}|S^j|$, and \emph{accept} if and only if $|\hat{J}|=k$.
\end{itemize}

The next claim says that with high probability, all influential variables of $g$ (the parity function closest to $f$) are inserted in $\hat{J}$.
\begin{claim} \label{clm:inf}
Let $f:\zo^n \to \zo$ be $1/10$-close to a parity function $g:\zo^n \to \zo$ with influential variables $J \subseteq [n]$.
With probability $1-o(1)$ the following conditions are simultaneously satisfied by the above test:
\begin{itemize}
\item $|S^j| > 100 \log k$ for every $j \in [n]$;
\item $|S^j_1| > \frac{3}{4}|S^j|$ for every $j \in J$.
\end{itemize}
\end{claim}
\begin{proof}
The expectation of $|S^j|$ is $1000\log k$, so using the Chernoff bound~\ref{eq:chernoff} we obtain the first item. 
Then use Claim~\ref{sec:noisy_inf_test} and another application of the Chernoff bound to get the second item.
The overall error probability is $o(1/k) = o(1)$.
\end{proof}

The next claim says that when $|J|\le k$, with high probability none of the non-influential
variables are inserted into $\hat{J}$. Before we proceed, let us call an index $i \in [q]$ {\em
  intersecting} with regard to~$J$ if $r_i \cap J \ne \emptyset$.
  Recall that the probability of any one element of $J$ belonging to $r_i$ is $\rho/k$; therefore
  the probability that $i$ is \emph{not} intersecting is $(1-\rho/k)^{|J|} \ge (1-\rho/k)^k \ge 1 -
  \rho$. If we set $\rho = 1/10$, the probability that $i$ is intersecting is at most $1/10$.
\begin{claim}\label{clm:noninf}
Let $f:\zo^n \to \zo$ be $1/10$-close to a parity function $g:\zo^n \to \zo$ with influential variables $J \subseteq [n]$.
If $|J|\le k$, then with probability $1-o(1)$ the following conditions are simultaneously satisfied by the above procedure:
\begin{itemize}
\item $|S^j| > 100 \log k$ for every $j \in [n]$;
\item for every $j \notin J$, the fraction of non-intersecting indices in $|S^j|$ is $>1/2$.
\end{itemize}
\end{claim}
Here, too, the proof follows by straightforward application of the Chernoff bound (and our
upper bound on $n$ as a function of~$k$).

To conclude the correctness of the tester, observe that by Claim \ref{clm:inf}, with high
probability $J \subseteq \hat{J}$, so if $J$ contains more than $k$ indices, then so will $\hat{J}$.
On the other hand, if $|J|\le k$ then by Claim~\ref{clm:noninf}, with high probability all sets
$S^j$ with $j \notin J$ contain a majority of non-intersecting indices. For each non-intersecting $i \in S^j$, it holds that $a_i$ is $0$ with probability $\ge 49/50$.
But in order for~$j$ to belong to $\hat{J}$, at least three quarters of the indices  $i \in
S^j$ must have $a_i = 1$, which implies that at least half of the non-intersecting indices
$i$ of $S^j$ must have $a_i = 1$. As there are at least $|S^j| / 2 > 50 \log k$ non-intersecting indices in
$S^j$, the Chernoff bound implies that this happens with probability at most $k^{-c}$ for some $c > 2$.
Since we are in the case $n=100k^2$, this probability is $o(1/n)$ for each $j\in[n]$, and we can apply the union bound to conclude that the error probability of the tester is $o(1)$.

Note that the test does more than testing: it actually \emph{identifies} the set $J$ of influential variables as long as it is of size $\le k$.

\subsection{Reducing the general case to $n=100k^2$}
Define a random $\ell$-way partition $\Pi$ of $[n]$ by assigning each $i \in [n]$ to one of 
$\ell$ classes uniformly at random (this may result in a partition into fewer than $\ell$ classes).
We can think of $\Pi$ as being constructed sequentially, say in increasing order of $i$.

Note that for any $S \subseteq[n]$, the probability that there there exist two distinct indices $i, j \in S$
assigned to the same class is at most ${ {|S| \choose 2} }/{\ell}$, by the union bound.
In particular, if $\ell = 100 k^2$ and $|S| \le 4k$, this probability is bounded above by $1/10$.

\begin{lemma} Let $k > 100$ and $n > 100k^2$. Given a subset $J \subseteq [n]$ and a $100k^2$-way
partition $\Pi = S_1,\ldots,S_{100k^2}$ of $[n]$, we denote by $N(\Pi,J)$ the number of classes
$S_i$ containing an odd number of elements from $J$. The following holds for randomly constructed
$100k^2$-way partitions $\Pi$:
\begin{itemize}
\item for each $J \subseteq [n]$ of size $|J| \leq k$, $\Pr_\Pi [N(\Pi,J) = |J|] > 9/10$,
\item for each $J \subseteq [n]$ of size $|J| > k$, $\Pr_\Pi [N(\Pi,J) > k] > 9/10$.
\end{itemize}
\end{lemma}
\begin{proof} 
If $|J| \leq k$, then with probability $\ge 9/10$ no pair of indices from $J$ belong to the same partition class, and hence $N(\Pi,J) = |J|$.

Now let $|J| > k$. Consider  the stage in the construction of the random partition $\Pi$ where all
but the last $k+1$ elements from $J$ were mapped to one of $\Pi$'s classes. If at this stage
$N(\Pi,J) >2k+2$ then we are done (since adding $k+1$ indices from $J$ can only change $N(\Pi,J)$ by
        $k+1$). Otherwise, we know that with probability at least $9/10$ no pair from a set of $\le 3k+2$
elements collides when randomly mapped to $100k^2$ classes. Therefore with probability at least
$9/10$, the last $k + 1$
elements are put into classes different from one another, and also into classes that so far contained an
even number of elements of~$J$.
\end{proof}

Once such a partition is obtained, we can simulate access to a function $f':\zo^{100k^2} \to \zo$ by querying $f$ on inputs that are constant within each partition class, and reduce the original problem of testing~$f$ to the problem of testing whether~$f'$ is a $k$-parity.

Putting everything together, we have proved our upper bound:
\begin{theorem}
There exists a non-adaptive tester that uses $O(k\log k)$ queries to a given function $f:\zo^n\rightarrow\zo$,
and decides with probability at least $2/3$ whether $f$ is in or far from $\PAR_k^n$.
\end{theorem}

\section{Lower bounds}\label{sec:lower}

\subsection{The one-way communication complexity of $k$-disjointness}

In two-party communication complexity~\cite{yao:distributive,comm_book}, two parties (Alice and Bob)
    have inputs~$x$ and $y$, respectively, and want to compute some function of $x$ and $y$.
Unlimited access to their respective inputs and arbitrary computations are allowed, and the measure  for the protocol's efficiency is the number
of bits of communication they need to transmit to each other. We consider the model where Alice and Bob share
a common source of randomness (``public coin'') and are allowed to err with probability at most $1/3$.

In the \emph{$k$-disjointness} problem, Alice and Bob receive two $k$-sets $x, y \in {[n] \choose k}$ and would like to
determine if $x \cap y = \emptyset$ or not. This problem is known to have communication complexity $\Theta(k)$. The upper bound is due to
H{\aa}stad and Wigderson~\cite{disjointnessUB}, the lower bound to Kalyanasundaram and Schnitger,
and subsequent simplifications and strengthenings were found by Razborov~\cite{disjointnessLB2} and Bar-Yossef et al.~\cite{disjointnessLB3}.
The H{\aa}stad-Wigderson protocol is interactive (i.e., it uses many rounds of communication), and we show here this is
actually necessary: if we just allow one-way communication from Alice to Bob, then the lower bound goes up from $\Omega(k)$ to $\Omega(k\log k)$ bits.
As mentioned in the introduction, this same bound was also independently observed in~\cite{patrascu:disjblog,dks:sparsedisj}.

\begin{theorem}\label{thm:onewaycclowerbound}
The one-way communication complexity of the $k$-disjointness problem is $\Theta(k \log k)$
for $k \le \sqrt{n/2}$, and $\Theta(\log{n\choose k})$ for $k>\sqrt{n/2}$.
The lower bound applies to the special case of \emph{unique} $k$-disjointness, where the inputs
satisfy either $x \cap y = \emptyset$ or $|x \cap y| = 1$.
\end{theorem}

\begin{proof}
{\bf Upper bound.}
First note that Alice can just send Bob the index of her input~$x$ in the set of all weight-$k$ strings of length~$n$, at the expense of $\log{n\choose k}$ bits. If $k\le\sqrt{n/2}$ then we can do something better, as follows.
Alice and Bob use the shared randomness to choose a random partition of their inputs into $b=O(k^2)$
buckets, each of size $n/b$. By standard birthday paradox arguments, with probability close
to~1 no two 1-positions in~$x$ will end up in the same bucket, and no two 1-positions in~$y$ will
end up in the same bucket. We assume in the remainder of the upper-bound argument that the bucketing
has succeeded in this sense. Note that $x$ and $y$ intersect iff there is an $i\in[b]$ such that Alice and Bob's
strings in the $i$th bucket are equal and non-zero. For each of her $k$ non-empty buckets, Alice
sends Bob the index of that bucket, and uses the well-known public-randomness equality protocol~\cite[Example~3.13]{comm_book} on
that bucket: they choose $2\log k$ uniformly random strings $r_1,\ldots,r_{2\log
  k}\in\zo^{n/b}$ and Alice sends over the inner products (mod~2) of her bucket with each of
  those strings. Bob compares the bits he received with the inner products of $r_1,\ldots,r_{2\log
    k}$ with his corresponding bucket. If their two buckets are the same then all inner
    products will be the same, and if their two buckets differ in at least one bit-position then
    they will see a difference in those inner products, except with probability $1/2^{2\log
      k}=1/k^2$. Bob checks whether one of Alice's non-empty buckets equals his corresponding bucket. If so he concludes that $x$ and~$y$ intersect, and otherwise he concludes that they are disjoint.  Taking the union bound over the probability that the bucketing fails and the probability that one of the $k$ equality tests fails, shows that the error probability is close to~0.  The communication cost of this one-way protocol is $O(\log k)$ bits for each of Alice's non-empty buckets, so $O(k\log k)$ bits in total.

{\bf Lower bound.}
First consider the case $k \le \sqrt{n/2}$. Let $x$ be Alice's input, viewed as an $n$-bit string of Hamming weight $k$.
For Alice we restrict our attention to inputs of a particular structure. Namely, partition $[n]$ into $k$ consecutive
sets of size $n / k \ge 2k$. The inputs we allow contain precisely one bit set to~1 inside each block of the partition, and
moreover the offset of the unique index set to one within the $i$th block is an integer in $\{0,\ldots,2k - 1\}$.
In this case, $x$ describes a message $M$ of $k$ integers $m_1,\dots,m_{k}$,
each in the interval $\{0,\dots,2k-1\}$. $M$ can also be viewed as an $m$-bit long message, where
$m=k\log(2k)$. We can write Alice's input as $x=u(m_1)\dots u(m_{k})$, where $u(m_i) \in \zo^{n/k}$ is the unary
expression of the number $m_i$ using $n/k$ bits (where the rightmost $n/k-k$ bits of each $u(m_i)$ are always zero).
For instance, the picture below illustrates the case where $n=40$, $k=4$, and $M=(1,7,0,5)$:
$$
x=
\overbrace{\underbrace{0100000000}_{u(m_1)}}^{n/k}\
\overbrace{\underbrace{0000000100}_{u(m_2)}}^{n/k}\
\overbrace{\underbrace{1000000000}_{u(m_3)}}^{n/k}\
\overbrace{\underbrace{0000010000}_{u(m_4)}}^{n/k}
$$
Let $\rho_x$ be the $q$-bit message that Alice sends on this input~$x$; this $\rho_x$ is a random variable, depending on the public coin.
Below we show that the message is a \emph{random-access code} for $M$, i.e., it allows a user to recover each bit of $M$
with probability at least $1-\delta$ (though not necessarily \emph{all} bits of $M$ simultaneously).
Then our lower bound will follow from  Nayak's random-access code lower bound~\cite{nayak:qfa}. This says that
$$
q\geq(1-H(\delta))m,
$$
where $\delta$ is the error probability of the protocol and $H(\delta)=-\delta\log(\delta)-(1-\delta)\log(1-\delta)$ is its binary entropy.

Suppose Bob is given $\rho_x$ and wants to recover some bit of~$M$.
Say this bit is the $\ell$th bit of the binary expansion of~$m_i$.
Then Bob completes the protocol using the following~$y$:
$y$ is 0 everywhere except on the $k$ bits in the $i$th  block of size $n/k$ whose offsets~$j$
(measured from the start of the block) satisfy the following:  $0 \le j < 2k$ and the $\ell$th bit of the binary expansion of~$j$ is~1.
The Hamming weight of $y$ is $k$ by definition.

Recall that Alice has a $1$ in block~$i$ only at position~$m_i$.
Hence $x$ and $y$ will intersect iff the $\ell$th bit of the binary expansion of $m_i$ is~1,
and moreover, the size of the intersection is either~0 or~1.
Running a $k$-disjointness protocol with success probability $1-\delta$ will now give Bob the
sought-for bit of $M$ with probability at least $1-\delta$, which shows that~$\rho_x$ is a random-access code for $M$.

If $k>\sqrt{n/2}$ then we can do basically the same lower-bound proof, except that the integers~$m_i$ are now
in the interval $\{0,\dots,n/k-1\}$, $m=k\log(n/k)$, and Bob puts only $n/2k<k$ ones in the $i$th block of~$y$
(he can put his remaining $k-n/2k$ indices somewhere at the end of the block, at an agreed place where Alice won't put 1s).
This gives a lower bound of $\Omega(k\log(n/k))=\Omega\big(\log{n\choose k}\big)$.
\end{proof}

We note that the lower bound holds even for \emph{quantum} one-way communication complexity, and even if we allow Alice and Bob to share unlimited quantum \emph{entanglement} before the protocol starts.  For the latter case Nayak's random access-code lower bound~\cite{nayak:qfa} needs to be replaced with Klauck's~\cite{klauck:qpcom} version, which is weaker by a factor of two.

\subsection{Non-adaptive lower bound for testing $k$-parities}

In a recent paper, Blais, Brody and Matulef \cite{bbm} made a clever connection between property testing and some well-studied problems in communication complexity. As one of the applications of this connection, they  used the $\Omega(k)$
lower bound for $k$-disjointness to prove an $\Omega(k)$ lower bound on testing whether a function is in or far from the class of $k$-parities.
We use their argument to get a better lower bound for \emph{non-adaptive} testers.

\begin{corollary}\label{coro_smp}
Let $1\leq k\leq n$.
If $k\leq\sqrt{n/2}$ then non-adaptive testers need at least $\Omega(k \log k)$ queries to test with success probability at least $2/3$
whether a given function $f:\zo^n\rightarrow\zo$ is in or far from $\PAR_k^n$;
and if $k>\sqrt{n/2}$ then they need at least $\Omega(\log{n\choose k})$ queries.
\end{corollary}

\begin{proof}
Let $k$ be even (a similar argument works for odd $k$).
Below we show how Alice and Bob can use a non-adaptive $q$-query tester for $k$-parities
to get a one-way public-coin communication complexity for $k/2$-disjointness with $q$ bits of communication.
The communication lower bound of Theorem~\ref{thm:onewaycclowerbound} then implies the result.

Alice forms the function $f = \chi_x$ and Bob forms the function $g = \chi_y$. Consider the function $h = \chi_{x \oplus y}$. Since
$|x \oplus y| = |x| + |y| - 2 |x \cap y|$ and $|x|+|y|=k$, the function $h$ is a $(k-2|x\cap y|)$-parity.
A $q$-query randomized tester is a probability distribution over $q$-query deterministic testers.
Alice and Bob use the public coin to jointly sample one of those deterministic testers.
Since the tester is non-adaptive, this fixes the $q$ queries that will be made.
For every such query $z \in \zo^n$, Alice sends Bob the bit $f(z)$.
This enables Bob to compute $h(z) = f(z) \oplus g(z)$ for all~$q$ queries and then to finish the
computation of the tester. Since an $\ell$-parity with $\ell<k$ has distance 1/2 from every $k$-parity, the tester will tell Bob whether $h$ is a $k$-parity or far from a $k$-parity; 
i.e., whether $x$ and $y$ intersect or not.
\end{proof}

As mentioned in the introduction, a lower bound for testing membership in $\PAR_k^n$ implies a lower bound for $\PAR_{\le k}^n$ and juntas.


\section{Conclusion and future work}

We end with a few comments and directions for future research:
\begin{itemize}
\item While our disjointness lower bound (Theorem~\ref{thm:onewaycclowerbound}) also applies to one-way \emph{quantum} protocols, our lower bound for testing (Corollary~\ref{coro_smp}) does not. The reason is that the overhead when turning a quantum tester into a communication protocol will be $O(n)$ qubits per query in the quantum case,
in contrast to the $O(1)$ bits per query in the classical case.
In fact, if $f$ is a $k$-parity then the Bernstein-Vazirani algorithm~\cite{bernstein_vazirani} finds $x$ itself using only one quantum query, so testing for $k$-parities is trivial for quantum algorithms.
\item For \emph{adaptive} testers for $k$-parities there is still a gap between the best lower bound of  $k-o(k)$ queries~\cite{blais&kane:klinearity},
and the best upper bound of $O(k \log k)$ queries. It would be interesting to close this gap.
\end{itemize}

\subsection*{Acknowledgements}
We thank Pavel Pudlak for asking a question (communicated via HB) related to lower bounds for monotone circuits and span programs, which prompted RdW to prove the lower bound part of Theorem~\ref{thm:onewaycclowerbound} in October 2006; this proof was unpublished until now. We thank Joshua Brody for referring us to~\cite{dks:sparsedisj}, Anirban Dasgupta for sending us a copy of that paper, David Woodruff for a pointer to~\cite{patrascu:disjblog}, and the anonymous referees for many useful comments and references.

\bibliographystyle{alpha}
\bibliography{fiso}

\end{document}